\documentclass[11pt,reqno]{amsart}
\usepackage{amsfonts, amsmath, amssymb, color}
\usepackage{amsthm}
\usepackage{booktabs}
\usepackage{enumitem}
\usepackage{float}
\usepackage{hhline}
\usepackage{soul, color, xcolor}
\usepackage{lipsum}
\usepackage{lscape}
\usepackage{subfig}
\usepackage{textcomp}
\usepackage{tikz}
\usetikzlibrary{decorations.pathmorphing,shapes,arrows,positioning}
\usepackage{xcolor}
\usepackage{young}
\usepackage{youngtab}
\usepackage{ytableau}
\usepackage{rotating}
\usepackage{graphicx}
\usepackage[colorlinks=true, linkcolor=blue, citecolor=blue, urlcolor=blue]{hyperref} %backref=page]{hyperref}

\hypersetup{colorlinks,linkcolor=blue,urlcolor=blue,citecolor=blue}

\usepackage{scalefnt}
\usetikzlibrary{calc}
\usepackage{a4wide}

\usepackage[numbers,sort&compress]{natbib}

\allowdisplaybreaks[4]
\theoremstyle{plain}
\newtheorem{thm}{Theorem}[section]
\newtheorem{lem}[thm]{Lemma}

\newtheorem{cor}[thm]{Corollary}
\newtheorem{rem}[thm]{Remark}

\theoremstyle{definition}

\newtheorem{exmp}[thm]{Example}

\numberwithin{equation}{section} \errorcontextlines=0

\makeatletter
\newcommand{\Rmnum}[1]{\expandafter\@slowromancap\romannumeral #1@}

\begin{document}
\title{Unified monogamy and polygamy relations for multipartite systems}
\author{Yue Cao}
\address{School of Arts and Sciences, Guangzhou Maritime University, Guangzhou, Guangdong 510725, China}
\email{caoyue@gzmtu.edu.cn}
\author{Naihuan Jing*}
\address{Department of Mathematics, North Carolina State University, Raleigh, NC 27695, USA}
\email{jing@ncsu.edu}
\author{Kailash Misra}
\address{Department of Mathematics, North Carolina State University, Raleigh, NC 27695, USA}
\email{misra@ncsu.edu}
\author{Yiling Wang}
\address{Department of Mathematics, North Carolina State University, Raleigh, NC 27695, USA}
\email{ywang327@ncsu.edu}
%Convex-roof extended negativity of assistance (CRENoA)
\thanks{Supported in part by Simons Foundation grant nos. 36482 and MP-TSM-00002518}
\thanks{$*$Corresponding author: jing@ncsu.edu}

\begin{abstract}

For a bipartite entanglement measure $\mathcal{E}$ that satisfies the $\gamma$th-power monogamy inequality (Eq.~\eqref{e:chap1-ineq1}), and for its assisted counterpart $\mathcal{E}_a$ that obeys the $\delta$th-power polygamy inequality (Eq.~\eqref{e:chap1-ineq2}), we introduce a unified, tunable framework indexed by a parameter $m\geq1$. Within this framework, we derive two hierarchical families of refined inequalities:
\begin{enumerate}[label={\bf(\roman*)}]
\item a tightened $\alpha$-power monogamy relation for $\mathcal{E}$, valid for all $\alpha \geq m\gamma$;
\item a tightened $\beta$-power polygamy relation for $\mathcal{E}_a$, applicable for $(m-1)\delta < \beta \leq m\delta$.
\end{enumerate}

As $m$ increases, the bounds become progressively tighter, recovering known results at $m=1$. Notably, the optimal monogamy bound emerges as a piecewise function of $\alpha$, with additional correction terms activated as $\alpha$ crosses successive integer thresholds, thereby offering a sharper characterization of entanglement distribution.

We demonstrate that our results generalize and strengthen existing monogamy and polygamy relations through analytical comparisons and numerical evaluations using concurrence and concurrence of assistance. This hierarchical, parameterized approach offers enhanced and flexible tools for applications in quantum communication, quantum networks, and multipartite quantum information processing.
\end{abstract}

\maketitle

\section{\textbf{Introduction}}
%Quantum entanglement is a cornerstone of quantum mechanics and a vital resource in quantum information processing, enabling protocols such as quantum teleportation, dense coding, and quantum key distribution [1¨C3]. A key property distinguishing quantum from classical correlations is the monogamy of entanglement (MoE), which restricts the shareability of entanglement among multiple subsystems [4]. This constraint plays a crucial role in securing quantum communication and understanding the structure of multipartite quantum states [5¨C7].

Quantum entanglement lies at the heart of quantum mechanics and plays a foundational role in quantum information theory and quantum communication protocols~\cite{CAF, HHHH, DFSC}. A key distinguishing feature of entanglement--unlike classical correlations--is its limited shareability among subsystems, a phenomenon known as the \emph{monogamy of entanglement} (MoE)~\cite{ASI}. This property imposes a fundamental constraint on how entanglement may be distributed across multiple parties and has critical implications for quantum key distribution, quantum networks, and various multipartite quantum technologies~\cite{B, CB, GR}.

In contrast, the dual notion of assisted entanglement, arising from the resource-theoretic perspective, leads to the concept of polygamy of entanglement. This captures the scenario in which entanglement can be spread among several subsystems with the help of an external assisting party~\cite{KDS1, K2, GG}, and has proved useful in understanding the structural dualities of quantum correlations in multipartite settings.

Let $\mathcal{E}$ be a bipartite entanglement measure. For an $N$-qubit state $\rho_{AB_1 \cdots B_{N-1}}$, the $\gamma$-power monogamy inequality is given by~\cite[Thm.~1, Def.~1]{GG}:
\begin{equation}
\mathcal{E}^{\gamma}(\rho_{A|B_1 \cdots B_{N-1}}) \geq \sum_{i=1}^{N-1} \mathcal{E}^{\gamma}(\rho_{AB_i}),
\label{e:chap1-ineq1}
\end{equation}
where $\gamma > 0$ is the smallest exponent for which the inequality holds. The reduced density matrix $\rho_{AB_i}$ is obtained by tracing out all other subsystems except $A$ and $B_i$. The monogamy inequality for squared concurrence was first introduced by Coffman, Kundu, and Wootters~\cite{CKW}, and later extended to general multipartite systems by Osborne and Verstraete~\cite{OV1}. Related extensions to other correlation measures, such as quantum discord and hybrid quantum-classical quantities, have also been developed~\cite{G1, JHC, JLLF, ZF, KPSS, KDS2, OF, OV2}.

Dually, the polygamy inequality for the $\delta$-power of an assisted entanglement measure $\mathcal{E}_a$ takes the form~\cite[Thm.~1, Def.~1]{G}:
\begin{equation}
\mathcal{E}_{a}^{\delta}(\rho_{A|B_1 \cdots B_{N-1}}) \leq \sum_{i=1}^{N-1} \mathcal{E}_{a}^{\delta}(\rho_{AB_i}),
\label{e:chap1-ineq2}
\end{equation}
where $\delta > 0$ denotes the maximal exponent ensuring the inequality's validity. This form of inequality was first demonstrated using the squared concurrence of assistance by Gour et al.~\cite{GMS}, and later extended to general multipartite systems~\cite{GBS}.

%\textcolor{red}
{While existing monogamy and polygamy inequalities provide foundational constraints, their bounds are often loose and fail to capture the hierarchical and asymmetric nature of entanglement distribution in many-body systems and quantum networks. To overcome these limitations, we propose a unified parameterized framework that directly incorporates hierarchy and asymmetry into the formulation of entanglement constraints. Specifically we introduce:
 \begin{itemize}
    \item \textbf{A hierarchy parameter $m$,} an integer that controls the depth of correlation layers considered. By increasing $m$, we activate higher-order correction terms, constructing a tower of progressively tighter inequalities that mirror the layered sharing of entanglement in complex systems.
    \item \textbf{State-dependent asymmetry parameters $\mu$ and $\nu$,} which quantify the observed imbalance in bipartite entanglement across different partitions. These parameters allow our bounds to adapt to the specific geometry of the quantum state or network.
    \end{itemize}
Within this framework, we establish a hierarchy of $\alpha$-power monogamy inequalities valid in the regime $\alpha \geq m\gamma$, as well as a family of refined $\beta$-power polygamy inequalities applicable for $(m-1)\delta < \beta \leq m\delta$. These refined bounds become progressively tighter as $m$ increases and recover the standard inequalities when $m = 1$. This approach not only yields mathematically sharper inequalities but also offers a more physically nuanced description of entanglement sharing, with significant potential value for quantum key distribution, network routing, and multipartite protocol design.}

In particular, the optimal monogamy bound obtained through our framework takes the form of a piecewise-defined function in $\alpha$, where each segment incorporates additional correction terms activated when $\alpha$ crosses successive thresholds $2m$ (with $m \in \mathbb{Z}_+$). This structure enables strictly sharper inequalities than previous continuous formulations, while retaining analytical tractability. Our results not only generalize but also significantly strengthen the best-known monogamy and polygamy relations in the literature, including those in~\cite{CJMW1, GYG}. Both analytical derivations and numerical evaluations under concrete entanglement measures, such as concurrence and concurrence of assistance (CoA), demonstrate the effectiveness and generality of our approach.

This paper is organized as follows. Section~\ref{s:prelim} introduces the necessary preliminaries and mathematical tools. In Section~\ref{s:monogamy}, we present the main results on monogamy inequalities and analyze their performance. Section~\ref{s:polygamy} discusses the corresponding polygamy relations. Technical proofs and supplementary material are provided in Section~\ref{s:appendix}.

\section{\textbf{Preliminaries}}\label{s:prelim}

We investigate several inequalities concerning the binomial function $(1 + t)^x$ over certain intervals. These results are foundational to our analysis of $\alpha$th-monogamy and $\beta$th-polygamy relations in multipartite quantum systems. Our aim is to compare and refine earlier inequalities that have been widely used in the study of such entanglement relations.

For any real number $x$, we define the binomial coefficient as
\[
\binom{x}{n} = \frac{x(x-1)\cdots(x-n+1)}{n!}, \quad \text{with} \quad \binom{x}{0} = 1.
\]

\begin{lem}\label{l:chap2-lem1}
Let $m \geq 1$ be a parameter and $k > 0$ a fixed constant. Then for all $t \geq k$, the following holds:

\begin{itemize}
    \item[\textbf{(1)}] If $x \geq m$, then
    \begin{equation}\label{e:chap2-ineq1}
    (1 + t)^x \geq t^x + \sum_{n=0}^{\lfloor m-1 \rfloor} \binom{x}{n} \left[ (1 + k)^{x - n} - k^{x - n} \right] (t - k)^n.
    \end{equation}

    \item[\textbf{(2)}] If $m - 1 < x \leq m$ $(m\in\mathbb{Z}_{+})$, we have
    \begin{equation}\label{e:chap2-ineq2}
    \begin{aligned}
    (1 + t)^x &\leq t^x + \sum_{n=0}^{ m - 1 } \binom{x}{n} \left[ (1 + k)^{x - n} - k^{x - n} \right] (t - k)^n.
    \end{aligned}
    \end{equation}
\end{itemize}
Here, for any real number $a$, $\lfloor a \rfloor$ denote the floor function.
\end{lem}

From~\cite[Lem. 2.1 (2.4)]{CJMW1}, we have the following inequality:
\begin{equation}\label{e:chap2-ineq3}
(1 + t)^x \geq t^x + (1 + k)^x - k^x + k^{-x} - t^{-x} + (2^{m - 1} - 2)x(t - k),
\end{equation}
for $x \geq m \geq 2$, where $m := 1 + \log_2(r + 2)$ and $r \geq 0$ as defined in~\cite[Lem. 2.1]{CJMW1}.

\bigskip

This leads to the following corollary:

\begin{cor}\label{c:chap2-cor1}
If $t \geq k \geq 1$ and $x \geq m \geq 2$, then
\begin{equation}\label{e:chap2-ineq4}
\sum_{n=1}^{\lfloor m - 1 \rfloor} \binom{x}{n} \left[ (1 + k)^{x - n} - k^{x - n} \right] (t - k)^n \geq k^{-x} - t^{-x} + (2^{m - 1} - 2)x(t - k).
\end{equation}
\end{cor}

\bigskip

The lower bound in inequality~\eqref{e:chap2-ineq1} provides a uniformly tighter (or at least equivalent) estimate of $(1 + t)^x$ compared to previously known results for all integer $m \geq 1$, as shown in the following cases:

\begin{enumerate}[label=\textbf{(\roman*)}]
\item When $m = 1$, the inequalities in~\cite[Lem. 1]{XZL} are special instances of Lemma~\ref{l:chap2-lem1}.

\item When $m = 2$, under the conditions $t \geq k \geq 1$ and $x \geq 2$, inequality~\eqref{e:chap2-ineq1} simplifies to
\begin{equation}\label{e:chap2-ineq5}
(1 + t)^x \geq t^x + (1 + k)^x - k^x + x \left[ (1 + k)^{x - 1} - k^{x - 1} \right](t - k).
\end{equation}
Meanwhile,~\cite[Eq. (2.2)]{CJMW1} provides
\begin{equation}\label{e:chap2-ineq6}
(1 + t)^x \geq t^x + (1 + k)^x - k^x + k^{-x} - t^{-x}.
\end{equation}
Thus, Corollary~\ref{c:chap2-cor1} demonstrates that inequality~\eqref{e:chap2-ineq5} yields a sharper lower bound than~\eqref{e:chap2-ineq6}.

\item For $m > 2$, Corollary~\ref{c:chap2-cor1} confirms that inequality~\eqref{e:chap2-ineq1} strictly improves upon the bound in~\eqref{e:chap2-ineq3}, which appears as equation (2.4) in Lemma 2.1 of~\cite{CJMW1}, under the same conditions $t \geq k \geq 1$ and $x \geq m > 2$.
\end{enumerate}

\begin{lem}\label{l:chap2-lem2}
Let $\{p_i\}_{i=1}^{N}$ be a non-increasing sequence of positive numbers, i.e., $p_i \geq p_{i+1}$ for all $1 \leq i < N$.

\begin{itemize}
    \item[\textbf{(1)}] For $x \geq m$ with $m \geq 1$, we have
    \begin{equation}\label{e:chap2-ineq7}
    \left( \sum_{i=1}^{N} p_i \right)^x \geq p_1^x + \sum_{l=2}^{N} p_l^x \left( \sum_{n=0}^{\lfloor m - 1 \rfloor} \binom{x}{n} \left( l^{x - n} - (l - 1)^{x - n} \right) (\tau_l - (l - 1))^n \right),
    \end{equation}

    \item[\textbf{(2)}] For $m - 1 < x \leq m$ with $m_{\in\mathbb{Z}_{+}}\geq1$, we have
    \begin{equation}\label{e:chap2-ineq8}
    \begin{aligned}
    \left( \sum_{i=1}^{N} p_i \right)^x &\leq p_1^x + \sum_{l=2}^{N} p_l^x \left( \sum_{n=0}^{ m - 1 } \binom{x}{n} \left( l^{x - n} - (l - 1)^{x - n} \right) (\tau_l - (l - 1))^n \right).
    \end{aligned}
    \end{equation}
\end{itemize}
Here, $\tau_l := \frac{p_1 + p_2 + \cdots + p_{l-1}}{p_l}$ for $l = 2, \dots, N$.
\end{lem}

\section{\textbf{Monogamy Relations for the $\alpha$-th Power of Entanglement Measures}}\label{s:monogamy}

Let $\rho = \rho_{AB_{1}\cdots B_{N-1}}$ be an $N$-partite quantum state defined on the Hilbert space $\mathcal{H}_A \otimes \mathcal{H}_{B_1} \otimes \cdots \otimes \mathcal{H}_{B_{N-1}}$. When the context is clear, we denote $\mathcal{E}_{(a)}(\rho_{AB_i}) := \mathcal{E}_{(a)AB_i}$ and $\mathcal{E}_{(a)}(\rho_{A|B_1\cdots B_{N-1}}) := \mathcal{E}_{(a)A|B_1\cdots B_{N-1}}$ for brevity.

\bigskip

A new class of monogamy relations for multipartite quantum systems follows directly from inequality~\eqref{e:chap2-ineq7}.

\begin{thm}\label{t:chap3-thm1}
Let $\mathcal{E}$ be a bipartite entanglement measure satisfying the $\gamma$-th power monogamy relation~\eqref{e:chap1-ineq1}, and let $\rho_{AB_1\cdots B_{N-1}}$ be any $N$-qubit quantum state. Suppose the set $\{\mathcal{E}_i := \mathcal{E}_{AB_{i'}}\}_{i=1}^{N-1}$ is arranged in descending order such that $\mathcal{E}_i^\gamma \geq \mathcal{E}_{i+1}^\gamma > 0$ for $i = 1,\dots,N-2$. Then for $\alpha \geq m\gamma$ $(m \geq 1)$, the following monogamy relation holds:
\begin{equation}\label{e:chap3-ineq1}
\mathcal{E}^{\alpha}_{A|B_1\cdots B_{N-1}} \geq \mathcal{E}_1^\alpha + \sum_{l=2}^{N-1} \left( \sum_{n=0}^{\lfloor m-1 \rfloor} \binom{\alpha/\gamma}{n} \left[l^{\frac{\alpha}{\gamma} - n} - (l-1)^{\frac{\alpha}{\gamma} - n}\right](\tau_l - (l-1))^n \right)\mathcal{E}_l^\alpha,
\end{equation}
where $\tau_l := \frac{\sum_{i=1}^{l-1} \mathcal{E}_i^\gamma}{\mathcal{E}_l^\gamma}$ for $l = 2, \dots, N-1$.
\end{thm}

%\textcolor{red}
{\textbf{Physical Interpretation.} The parameter $m$ in \eqref{e:chap3-ineq1} selects the correlation order included in the bound. With $m=1$, we recover the basic monogamy inequality. Increasing $m$ adds higher-order correction terms that capture deeper sharing constraints among subsystems. Therefore, a larger $m$ yields a more complete description of the entanglement structure. This is particularly important for analyzing systems with strong multipartite correlations or for cryptographic applications that demand tighter bounds.}

\begin{lem}\label{l:chap3-lem1}
Let $\mathcal{E}$ be a monogamous entanglement measure satisfying~\eqref{e:chap1-ineq1}. For any tripartite quantum state $\rho_{ABC}$, there exists $\mu \geq 1$ such that
\begin{equation}\label{e:chap3-ineq2}
\mathcal{E}_{A|BC}^\gamma \geq \mu\,\mathcal{E}_{AB}^\gamma + \mathcal{E}_{AC}^\gamma \geq \mathcal{E}_{AB}^\gamma + \mathcal{E}_{AC}^\gamma,
\end{equation}
where $\mu$ depends on the specific state $\rho_{ABC}$.
\end{lem}

\begin{proof}
Since $\mathcal{E}$ is non-increasing under partial trace, we have $\mathcal{E}_{A|BC}^\gamma \geq \max\{\mathcal{E}_{AB}^\gamma, \mathcal{E}_{AC}^\gamma\}$. Assume $\mathcal{E}_{AB}^\gamma = \max\{\mathcal{E}_{AB}^\gamma, \mathcal{E}_{AC}^\gamma\}$.
\begin{itemize}
    \item[\textbf{(1)}]  If $\mathcal{E}_{AB}^\gamma = 0$, set $\mu = 1$, and~\eqref{e:chap3-ineq2} holds trivially.
    \item[\textbf{(2)}]  If $\mathcal{E}_{AB}^\gamma > 0$, set $\mu := \frac{\mathcal{E}_{A|BC}^\gamma - \mathcal{E}_{AC}^\gamma}{\mathcal{E}_{AB}^\gamma}$, then the equality $\mathcal{E}_{A|BC}^\gamma = \mu \mathcal{E}_{AB}^\gamma + \mathcal{E}_{AC}^\gamma$ holds, and~\eqref{e:chap1-ineq1} guarantees $\mu \geq 1$.
\end{itemize}
\end{proof}

\bigskip
The following corollary follows immediately from Theorem~\ref{t:chap3-thm1} and Lemma~\ref{l:chap3-lem1}.

\begin{cor}\label{c:chap3-cor1}
Let $\mathcal{E}$ be a bipartite entanglement measure satisfying the $\gamma$-th monogamy inequality~\eqref{e:chap1-ineq1} for any tripartite quantum state $\rho_{ABC}$. If $\mathcal{E}_{AB}^\gamma \geq \mathcal{E}_{AC}^\gamma > 0$, then for $\alpha \geq m\gamma$ $(m \geq 1)$ and $\mu \geq 1$,
\begin{equation}\label{e:chap3-ineq3}
\begin{aligned}
\mathcal{E}_{A|BC}^{\alpha} \geq \mu^{\frac{\alpha}{\gamma}} \mathcal{E}_{AB}^{\alpha} + \sum_{n=0}^{\lfloor m - 1 \rfloor} \binom{\alpha/\gamma}{n} \left[(1 + \mu)^{\frac{\alpha}{\gamma} - n} - \mu^{\frac{\alpha}{\gamma} - n} \right] \mu^n \left(\frac{\mathcal{E}_{AB}^\gamma}{\mathcal{E}_{AC}^\gamma} - 1\right)^n \mathcal{E}_{AC}^{\alpha}.
\end{aligned}
\end{equation}
\end{cor}
More generally, we have
\begin{thm}\label{t:chap3-thm2}
Let $\mathcal{E}$ be a bipartite entanglement measure satisfying the $\gamma$-th monogamy relation~\eqref{e:chap1-ineq1} for any any $N$-qubit quantum state $\rho_{AB_1\cdots B_{N-1}}$. Suppose the sequence $\{\mathcal{E}_i := \mathcal{E}_{AB_{i'}}\}_{i=1}^{N-1}$ is arranged in descending order such that $\mathcal{E}_i^\gamma \geq \mathcal{E}_{i+1}^\gamma > 0$ for $i = 1,\dots,N-2$.  Then for $\alpha \geq m\gamma$ $(m \geq 1)$, we have
\begin{equation}\label{e:chap3-ineq4}
\begin{aligned}
&\mathcal{E}^{\alpha}_{A|B_{1}\cdots B_{N-1}} \geq \mu^{\frac{\alpha}{\gamma}}\left(\mathcal{E}_{1}^{\alpha}+\sum_{l=2}^{N-2}\left[\sum_{n=0}^{\lfloor m-1\rfloor}\binom{\alpha/\gamma}{n}(l^{\frac{\alpha}{\gamma}-n}-(l-1)^{\frac{\alpha}{\gamma}-n})(\tau_{l}-(l-1))^{n}\right]\mathcal{E}^{\alpha}_{l}\right)\\
&+\left(\sum_{n=0}^{\lfloor m-1\rfloor}\binom{\alpha/\gamma}{n}\left[(1+\mu(N-2))^{\frac{\alpha}{\gamma}-n}-(\mu(N-2))^{\frac{\alpha}{\gamma}-n}\right]\mu^{n}\left[\tau_{N-1}-(N-2)\right]^{n}\right)\mathcal{E}^{\alpha}_{N-1}
\end{aligned}
\end{equation}
where $\mu \geq 1$ and $\tau_l := \frac{\sum_{i=1}^{l-1} \mathcal{E}_i^\gamma}{\mathcal{E}_l^\gamma}$ for $l = 2, \dots, N-1$.
\end{thm}

\begin{proof}
By the $\gamma$-monogamy relation~\eqref{e:chap1-ineq1}, we have
$\mathcal{E}^{\gamma}_{A|B_{1}\cdots B_{N-1}}\geq \mathcal{E}^{\gamma}_{1}+\mathcal{E}^{\gamma}_{2}+\cdots+\mathcal{E}^{\gamma}_{N-1}
$.
Since $\mathcal{E}_i^\gamma \geq \mathcal{E}_{i+1}^\gamma$ for $i=1,\dots,N-2$, there exists $\mu \geq 1$ such that
$
\mathcal{E}_{A|B_1\cdots B_{N-1}}^\gamma \geq \mu\sum_{i=1}^{N-2} \mathcal{E}_i^\gamma + \mathcal{E}_{N-1}^\gamma.
$
The inequality~\eqref{e:chap3-ineq4} then follows from Theorem~\ref{t:chap3-thm1} and Lemma~\ref{l:chap3-lem1}.
\end{proof}

\begin{rem}\label{r:chap3-rem1}
Let $\rho_{AB_1\cdots B_{N-1}}$ be an $N$-partite quantum state.
\begin{itemize}
    \item[\textbf{(1)}] If $\mu = 1$, then~\eqref{e:chap3-ineq4} reduces to~\eqref{e:chap3-ineq1}.
    \item[\textbf{(2)}] If $\mu > 1$, then the bound in~\eqref{e:chap3-ineq4} is strictly tighter than that in~\eqref{e:chap3-ineq1}, since the function $h(x,y) = (1+y)^x - y^x$ is strictly increasing in $y$ for $x \geq 1$ and $y > 0$.
\end{itemize}
\end{rem}

%\textcolor{red}
{\textbf{Physical Interpretation.} Theorem \ref{t:chap3-thm2} simultaneously incorporates the hierarchy parameter $m$ and the state-dependent asymmetry parameter $\mu$ (defined in Lemma \ref{l:chap3-lem1}). When $\mu > 1$, the entanglement distribution is imbalanced, and the bound in \eqref{e:chap3-ineq4} is strictly tighter than the general bound in \eqref{e:chap3-ineq1}. This reflects the physical reality that entanglement distribution in many-body systems is typically non-uniform; our framework adapts to this asymmetry to provide more precise constraints. For instance, in quantum networks, entanglement distribution is often asymmetric due to channel losses or topological constraints, and our bounds can provide more accurate estimates for entanglement allocation in such networks.}

\bigskip
{\bf Comparison of the monogamy relations for entanglement measure $\mathcal{E}$}. By Remark~\ref{r:chap3-rem1} and \cite[Cor.~3.8]{CJMW1}, the following unified family of monogamy inequalities for the $\alpha$-th power of an entanglement measure $\mathcal{E}$ holds:

\begin{align*}
\mathcal{E}^{\alpha}_{A|B_{1}\cdots B_{N-1}} &\geq \mu^{\frac{\alpha}{\gamma}}\left(\mathcal{E}_{1}^{\alpha}+\sum_{l=2}^{N-2}\left[\sum_{n=0}^{\lfloor m-1\rfloor}\binom{\alpha/\gamma}{n}(l^{\frac{\alpha}{\gamma}-n}-(l-1)^{\frac{\alpha}{\gamma}-n})(\tau_{l}-(l-1))^{n}\right]\mathcal{E}^{\alpha}_{l}\right)\\
+&\left(\sum_{n=0}^{\lfloor m-1\rfloor}\binom{\alpha/\gamma}{n}\left[(1+\mu(N-2))^{\frac{\alpha}{\gamma}-n}-(\mu(N-2))^{\frac{\alpha}{\gamma}-n}\right]\mu^{n}\left[\tau_{N-1}-(N-2)\right]^{n}\right)\mathcal{E}^{\alpha}_{N-1}
\\ &\geq\mathcal{E}_{1}^{\alpha}+\sum_{l=2}^{N-1}\left(\sum_{n=0}^{\lfloor m-1\rfloor}\binom{\alpha/\gamma}{n}(l^{\frac{\alpha}{\gamma}-n}-(l-1)^{\frac{\alpha}{\gamma}-n})(\tau_{l}-(l-1))^{n}\right)\mathcal{E}^{\alpha}_{l} \\
&\geq \mathcal{E}_{1}^{\alpha}+\sum_{l=2}^{N-1}\left[l^{\frac{\alpha}{\gamma}}-(l-1)^{\frac{\alpha}{\gamma}})+(l-1)^{-\frac{\alpha}{\gamma}}-\tau_{l}^{-\frac{\alpha}{\gamma}}+\left(2^{m-1}-2 \right)\left[\tau_{l}-(l-1)\right]\right]\mathcal{E}^{\alpha}_{l} \\
&\geq \mathcal{E}_{1}^{\alpha}+\sum_{l=2}^{N-1}\left[l^{\frac{\alpha}{\gamma}}-(l-1)^{\frac{\alpha}{\gamma}}\right]\mathcal{E}^{\alpha}_{l} \geq \mathcal{E}_{1}^{\alpha}+\left(2^{\frac{\alpha}{\gamma}}-1\right)\sum_{l=2}^{N-1}\mathcal{E}^{\alpha}_{l}\geq \sum_{l=1}^{N-1}\mathcal{E}_{l}^{\alpha}
\end{align*}
for all $\alpha \geq m\gamma$ with $m \geq 2$, where $\mu \geq 1$ and $\tau_l = \frac{\sum_{i=1}^{l-1} \mathcal{E}_i^{\gamma}}{\mathcal{E}_l^{\gamma}}$ for $l = 2, \dots, N-1$.

Therefore, for $\alpha \geq m\gamma$ $(m \geq 2)$, the lower bound given in~\eqref{e:chap3-ineq4} is sharper than that of~\eqref{e:chap3-ineq1}. Furthermore, by Corollary~\ref{c:chap2-cor1}, the bound in~\eqref{e:chap3-ineq1} improves upon the earlier results in \cite[Lem.~2.1 (2.4)]{CJMW1} and thereby also those in \cite{GYG}.

\bigskip

We now demonstrate that the $\alpha$-th $(0\leq \alpha \leq \gamma)$ power monogamy relation derived here outperforms existing results using the example of concurrence.

Recall that the concurrence for a pure state $\rho_{AB} \in \mathcal{H}_A \otimes \mathcal{H}_B$ is defined as \cite{U,RBCGM}:
$ C(|\psi\rangle_{AB}) = \sqrt{2[1 - \operatorname{Tr}(\rho_A^2)]} = \sqrt{2[1 - \operatorname{Tr}(\rho_B^2)]}, $
where $\rho_A$ (resp. $\rho_B$) is the reduced density matrix by tracing over the subsystem $B$ (resp. $A$). For a mixed state $\rho_{AB}$, the concurrence is given by the convex roof extension \cite{YS}:
$ C(\rho_{AB}) = \min_{\{p_i, |\psi_i\rangle\}} \sum_i p_i C(|\psi_i\rangle), $
where the minimum is taken over all possible pure decompositions $\rho_{AB} = \sum_i p_i |\psi_i\rangle\langle\psi_i|$ with $p_i \geq 0$ and $\sum_i p_i = 1$.

The following result follows directly from Theorem~\ref{t:chap3-thm2}.

\begin{cor} \label{c:chap3-cor2}
Let $C$ be the concurrence entanglement measure satisfying the second-order monogamy relation~\eqref{e:chap1-ineq1}, and let $\rho_{AB_1\cdots B_{N-1}}$ be any $N$-qubit quantum state. Assume the sequence $\{C_i = C_{AB_{i'}}\}_{i=1}^{N-1}$ is arranged in descending order, such that $C_i^2 \geq C_{i+1}^2 > 0$ for $i = 1, \dots, N-2$. Then, for all $\alpha \geq 2m$ with $m \geq 1$, we have
\begin{equation}\label{e:chap3-ineq5}
\begin{aligned}
&C^{\alpha}_{A|B_{1}\cdots B_{N-1}} \geq \mu^{\frac{\alpha}{2}}\left(C_{1}^{\alpha}+\sum_{l=2}^{N-2}\left[\sum_{n=0}^{\lfloor m-1\rfloor}\binom{\alpha/2}{n}(l^{\frac{\alpha}{2}-n}-(l-1)^{\frac{\alpha}{2}-n})(\tau_{l}-(l-1))^{n}\right]C^{\alpha}_{l}\right)\\
&+\left(\sum_{n=0}^{\lfloor m-1\rfloor}\binom{\alpha/2}{n}\left[(1+\mu(N-2))^{\frac{\alpha}{2}-n}-(\mu(N-2))^{\frac{\alpha}{2}-n}\right]\mu^{n}\left[\tau_{N-1}-(N-2)\right]^{n}\right)C^{\alpha}_{N-1}
\end{aligned}
\end{equation}
where $\mu \geq 1$ and $\tau_l = \frac{\sum_{i=1}^{l-1} C_i^2}{C_l^2}$ for $l = 2, \dots, N-1$ with $N \geq 4$.
\end{cor}

\begin{exmp}
Let $\rho = |\psi\rangle\langle\psi|$ be the three-qubit state defined in \cite{AACJLT}:
\[ |\psi\rangle = \lambda_0 |000\rangle + \lambda_1 e^{i\varphi} |100\rangle + \lambda_2 |101\rangle + \lambda_3 |110\rangle + \lambda_4 |111\rangle, \]
where $\sum_{i=0}^4 \lambda_i^2 = 1$ and $\lambda_i \geq 0$ for all $i$. Then the concurrences are given by: $C_{A \mid BC}=2 \lambda_0 \sqrt{\lambda_2^2+\lambda_3^2+\lambda_4^2}$, $C_{AB}=2 \lambda_0 \lambda_2$, and $C_{AC}=2 \lambda_0 \lambda_3$.
Set $\lambda_0 = \lambda_1 = \lambda_2 = \frac{1}{2}$ and $\lambda_3 = \lambda_4 = \frac{\sqrt{2}}{4}$. Then we compute:
$C_{A|BC} = \frac{\sqrt{2}}{2}$, $C_{AB} = \frac{1}{2}$, and $C_{AC} = \frac{\sqrt{2}}{4}$.
By Lemma~\ref{l:chap3-lem1} and Corollary~\ref{c:chap3-cor1}, we obtain:$\mu = \frac{C_{A|BC}^2 - C_{AC}^2}{C_{AB}^2} = \frac{3}{2}$, and $t = \mu \cdot \frac{C_{AB}^2}{C_{AC}^2} = 3. $

Using Corollary~\ref{c:chap3-cor2}, for any $\alpha \geq 2$, the lower bound from our $\alpha$th-monogamy relation becomes
\begin{align*}
Z_1 &= \mu^{\frac{\alpha}{2}} C_{AB}^{\alpha} +
\sum_{n=0}^{\lfloor \alpha/2 - 1 \rfloor} \binom{\alpha/2}{n}
\left[(1+\mu)^{\frac{\alpha}{2} - n} - \mu^{\frac{\alpha}{2} - n} \right]
\mu^n \left(\frac{C_{AB}^2}{C_{AC}^2} - 1\right)^n C_{AC}^{\alpha} \\
&= \left(\frac{3}{2}\right)^{\frac{\alpha}{2}}\left(\frac{1}{2}\right)^{\alpha} +
\sum_{n=0}^{\lfloor \alpha/2 - 1 \rfloor} \binom{\alpha/2}{n}
\left[\left(\frac{5}{2}\right)^{\frac{\alpha}{2} - n} - \left(\frac{3}{2}\right)^{\frac{\alpha}{2} - n}\right]
\left(\frac{3}{2}\right)^n \left(\frac{\sqrt{2}}{4}\right)^{\alpha}.
\end{align*}

The corresponding bound from Theorem~\ref{t:chap3-thm1} is
\begin{align*}
Z_2 &= C_{AB}^{\alpha} +
\sum_{n=0}^{\lfloor \alpha/2 - 1 \rfloor} \binom{\alpha/2}{n}
\left(2^{\frac{\alpha}{2} - n} - 1\right) \left(\frac{C_{AB}^2}{C_{AC}^2} - 1\right)^n C_{AC}^{\alpha} \\
&= \left(\frac{1}{2}\right)^{\alpha} +
\sum_{n=0}^{\lfloor \alpha/2 - 1 \rfloor} \binom{\alpha/2}{n}
\left(2^{\frac{\alpha}{2} - n} - 1\right) \left(\frac{\sqrt{2}}{4}\right)^{\alpha}.
\end{align*}

For any $\alpha\geq 4$, the $\alpha$th-monogamy relation given in \cite[Thm. 3.7]{CJMW1} is
\begin{align*}
Z_3 &= C_{AB}^{\alpha} + \left(2^{\frac{\alpha}{2}} - \left(\frac{C_{AB}^2}{C_{AC}^2}\right)^{-\frac{\alpha}{2}} + \left(2^{\frac{\alpha}{2}-1} - 2\right)\frac{\alpha}{2} \left(\frac{C_{AB}^2}{C_{AC}^2} - 1\right)\right) C_{AC}^{\alpha} \\
&= \left(\frac{1}{2}\right)^{\alpha} + \left(2^{\frac{\alpha}{2}} - 2^{-\frac{\alpha}{2}} + \left(2^{\frac{\alpha}{2}-1} - 2\right)\frac{\alpha}{2}\right) \left(\frac{\sqrt{2}}{4}\right)^{\alpha}.
\end{align*}

Bounds from \cite{JLLF,ZF,JF1} respectively are
\begin{align*}
    Z_{_4}=(\frac{1}{2})^{\alpha}+\left(2^{\frac{\alpha}{2}}-1\right)(\frac{\sqrt{2}}{4})^{\alpha},~~
    Z_{_5}= (\frac{1}{2})^{\alpha}+\frac{\alpha}{2}(\frac{\sqrt{2}}{4})^{\alpha},~~
    Z_{_6}= (\frac{1}{2})^{\alpha}+(\frac{\sqrt{2}}{4})^{\alpha}.
\end{align*}
\end{exmp}

\begin{figure}[H]
    \centering
    \includegraphics[width=0.8\textwidth]{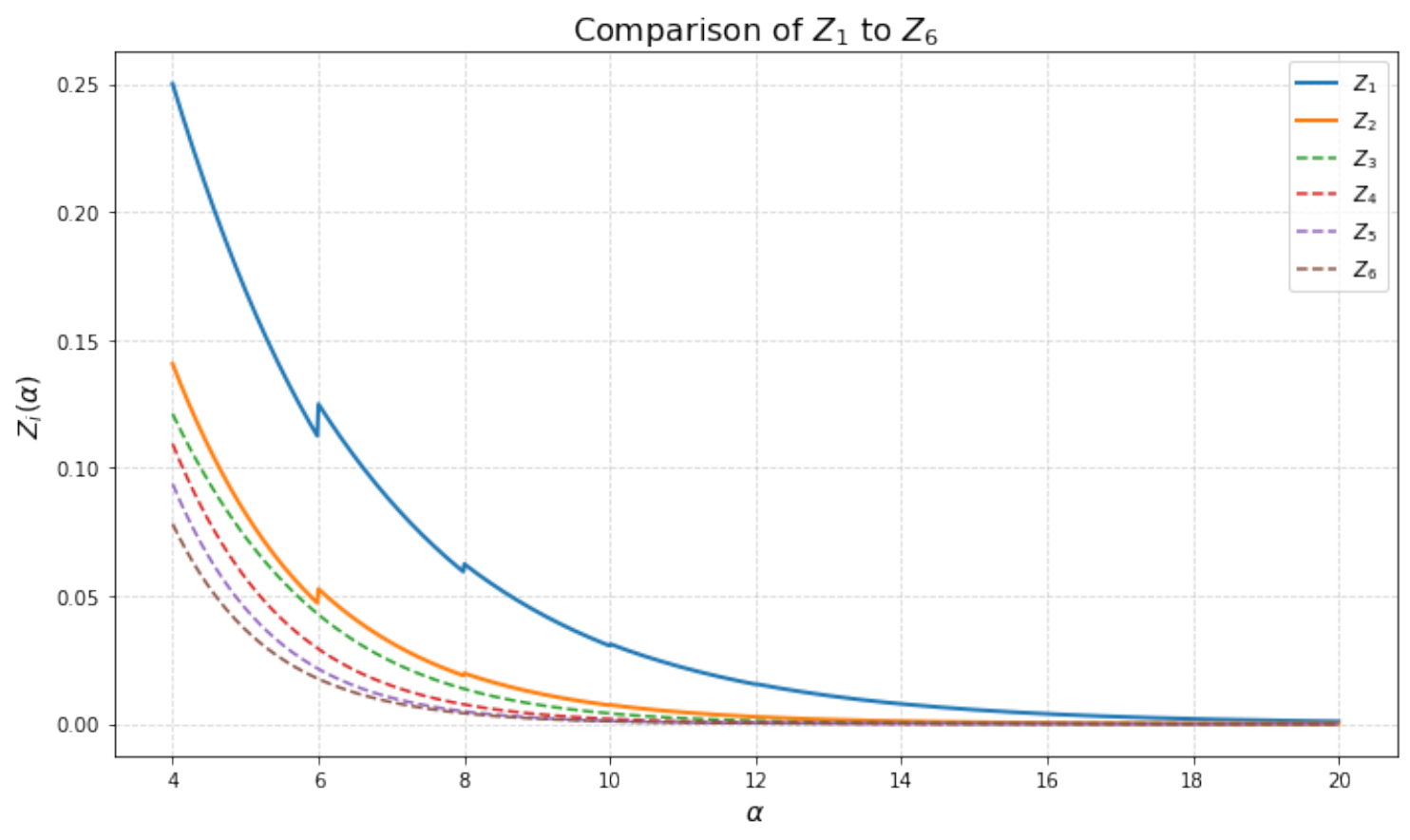}
    \caption{
    Comparison of six lower bounds $Z_i$ $(i = 1,\dots,6)$ for the $\alpha$-th power monogamy relation of concurrence under the three-qubit state from Example~\ref{c:chap3-cor2}. Bounds $Z_1$ and $Z_2$ are derived via piecewise expansion based on Corollary~\ref{c:chap3-cor2} and Theorem~\ref{t:chap3-thm1}, respectively, and exhibit a staircase structure with respect to $\alpha$. %\textcolor{red}
    {This stepwise structure reflects the piecewise refinement induced by increasing $m$: each ``step'' corresponds to activation of a higher-order correction term when $\alpha$ crosses $2\gamma$, mirroring the discrete nature of entanglement sharing in multi-qubit systems. Bounds $Z_3$ to $Z_6$ correspond to previously proposed results. Among them, $Z_1$ (which incorporates the state-dependent asymmetry parameter) provides the tightest lower bound for all $\alpha \geq 4$, demonstrating the efficacy of high-order refinement and the importance of accounting for entanglement asymmetry. The overall downward trend illustrates how the monogamy of entanglement strengthens as $\alpha$ increases.}
    }
    \label{fig:pZ123456}
\end{figure}

\section{\textbf{Polygamy relation for the $\beta$th power of assisted entanglement }}\label{s:polygamy}

The following Theorem is directly derived from \eqref{e:chap2-ineq8}.
\begin{thm}\label{t:chap4-thm1}
Let $\mathcal{E}$ be a bipartite assisted entanglement measure satisfying the $\delta$th-polygamy \eqref{e:chap1-ineq2} and $\rho_{AB_{1}\cdots B_{N-1}}$ any $N$-qubit quantum state. Arrange $\{\mathcal{E}_{a_{i}}=\mathcal{E}_{aAB_{i'}}|i=1,\cdots,N-1\}$ in descending order with $\mathcal{E}^{\delta}_{a_{i}}\geq \mathcal{E}^{\delta}_{a_{i+1}}>0$ for $i=1,\cdots,N-2$, then
\begin{equation}\label{e:chap4-ineq1}
\begin{aligned}
\mathcal{E}^{\beta}_{aA|B_{1}\cdots B_{N-1}} \leq \mathcal{E}_{a_{1}}^{\beta}+\sum_{l=2}^{N-1}\mathcal{E}^{\beta}_{a_{l}}\left(\sum_{n=0}^{ m-1}\binom{\beta/\delta}{n}(l^{\frac{\beta}{\delta}-n}-(l-1)^{\frac{\beta}{\delta}-n})(\tau_{a_{l}}-(l-1))^{n}\right)
\end{aligned}
\end{equation}
for $(m-1)\delta<\beta\leq  m\delta$ $(m_{\in\mathbb{Z}_{+}}\geq1)$, where $\tau_{a_l} = \frac{\sum_{i=1}^{l-1} \mathcal{E}_{a_i}^\delta}{\mathcal{E}_{a_l}^\delta}$ for $l = 2, 3, \dots, N-1$.
%\begin{equation}\label{e:chap4-ineq2}
%\begin{aligned}
%\mathcal{E}^{\beta}_{aA|B_{1}\cdots B_{N-1}} \leq \mathcal{E}_{a_{1}}^{\beta}+\sum_{l=2}^{N-1}\mathcal{E}^{\beta}_{a_{l}}\left(\sum_{n=0}^{\lceil m-1\rceil}\binom{\beta/\delta}{n}(l^{\frac{\beta}{\delta}-n}-(l-1)^{\frac{\beta}{\delta}-n})(\tau_{a_{l}}-(l-1))^{n}\right)
%\end{aligned}
%\end{equation}
%for $\lfloor m\rfloor\delta<\beta\leq m\delta$ $(m\geq1)$, where $\tau_{a_l} = \frac{\sum_{i=1}^{l-1} \mathcal{E}_{a_i}^\delta}{\mathcal{E}_{a_l}^\delta}$ for $l = 2, 3, \dots, N-1$.
\end{thm}

%\textcolor{red}
{\textbf{Physical Interpretation.} In the polygamy inequality \eqref{e:chap4-ineq1}, the parameter $m$ plays a role analogous to that in the monogamy case: it controls the depth of correlations considered in the upper bound. For $\beta$ within the interval $((m-1)\delta, m\delta]$, the bound incorporates up to $m-1$ correction terms. This enables a finer-grained description of the sharability of assisted entanglement, which is particularly relevant in quantum networks where entanglement assistance is often utilized to distribute correlations among multiple parties.}

Similar to Lemma \ref{l:chap3-lem1}, we have
\begin{lem}\label{l:chap4-lem1} For a tripartite quantum state $\rho_{ABC}$, there exist $0\leq\nu\leq1$ satisfies
\begin{equation}\label{e:chap4-ineq3}
\begin{aligned}
\mathcal{E}^{\delta}_{a_{A|BC}}\leq \nu\mathcal{E}^{\delta}_{a_{AB}}+\mathcal{E}^{\delta}_{a_{AC}}\leq \mathcal{E}^{\delta}_{a_{AB}}+\mathcal{E}^{\delta}_{a_{AC}},
\end{aligned}
\end{equation}
where $\nu$ is related to the quantum state $\rho_{ABC}$, and $\mathcal{E}^{\delta}_{a_{AB}}=\max\{\mathcal{E}^{\delta}_{a_{AB}},\mathcal{E}^{\delta}_{a_{AC}}\}$.
\end{lem}
\begin{proof}
 Since assisted entanglement measure $\mathcal{E}_{a}$ is non-increasing under partial traces, we have $\mathcal{E}^{\gamma}_{a_{A|BC}}\geq \max\{\mathcal{E}^{\gamma}_{a_{AB}},\mathcal{E}^{\gamma}_{a_{AC}}\}$. Suppose $\mathcal{E}^{\gamma}_{a_{AB}}=\max\{\mathcal{E}^{\gamma}_{a_{AB}},\mathcal{E}^{\gamma}_{a_{AC}}\}$. {\bf(1)} If $\mathcal{E}^{\gamma}_{aAB}=0$, then suppose $\nu=1$. {\bf(2)} If $\mathcal{E}^{\gamma}_{a_{AB}}>0$, set $\nu=\frac{\mathcal{E}^{\gamma}_{a_{A|BC}} -\mathcal{E}^{\gamma}_{a_{AC}} }{\mathcal{E}^{\gamma}_{a_{AB}}}$, then we have $\mathcal{E}^{\gamma}_{a_{A|BC}}= \nu\mathcal{E}^{\gamma}_{a_{AB}}+\mathcal{E}^{\gamma}_{a_{AC}}$. And by \eqref{e:chap1-ineq2}, we get $0\leq\nu\leq1$.
\end{proof}

\begin{thm}\label{t:chap4-thm2}
Let $\mathcal{E}$ be a bipartite assisted entanglement measure satisfying the $\delta$th-polygamy \eqref{e:chap1-ineq2} and $\rho_{AB_{1}\cdots B_{N-1}}$ any $N$-qubit quantum state. Arrange $\{\mathcal{E}_{a_{i}}=\mathcal{E}_{aAB_{i'}}|i=1,\cdots,N-1\}$ in descending order with $\mathcal{E}^{\delta}_{a_{i}}\geq \mathcal{E}^{\delta}_{a_{i+1}}>0$ for $i=1,\cdots,N-2$, then
\begin{equation}\label{e:chap4-ineq4}
\begin{aligned}
&\mathcal{E}^{\beta}_{aA|B_{1}\cdots B_{N-1}} \leq \nu^{\frac{\beta}{\delta}}\left( \mathcal{E}_{a_{1}}^{\beta}+\sum_{l=2}^{N-2}\left[\sum_{n=0}^{ m-1}\binom{\beta/\delta}{n}(l^{\frac{\beta}{\delta}-n}-(l-1)^{\frac{\beta}{\delta}-n})(\tau_{a_{l}}-(l-1))^{n}\right]\mathcal{E}^{\beta}_{a_{l}}\right)\\
&+\left(\sum_{n=0}^{\lceil m-1\rceil}\binom{\beta/\delta}{n}\left[(1+\nu(N-2))^{\frac{\beta}{\delta}-n}-(\nu(N-2))^{\frac{\beta}{\delta}-n}\right]\nu^{n}\left[\tau_{a_{N-1}}-(N-2)\right]^{n}\right)\mathcal{E}^{\beta}_{a_{N-1}}
\end{aligned}
\end{equation}
for $(m-1)\delta<\beta\leq  m\delta$ $(m_{\in\mathbb{Z}_{+}}\geq1)$, where $0\leq\nu\leq1$ and $\tau_{a_{l}}=\frac{\mathcal{E}^{\delta}_{a_{1}}+\mathcal{E}^{\delta}_{a_{2}}+\cdots+\mathcal{E}^{\delta}_{a_{l-1}}}{\mathcal{E}^{\delta}_{a_{l}}},l=2,3,\cdots,N-1$ $(N\geq4)$.
%\begin{equation}\label{e:chap4-ineq5}
%\begin{aligned}
%&\mathcal{E}^{\beta}_{aA|B_{1}\cdots B_{N-1}} \leq \nu^{\frac{\beta}{\delta}}\left( \mathcal{E}_{a_{1}}^{\beta}+\sum_{l=2}^{N-2}\left[\sum_{n=0}^{\lceil m-1\rceil}\binom{\beta/\delta}{n}(l^{\frac{\beta}{\delta}-n}-(l-1)^{\frac{\beta}{\delta}-n})(\tau_{a_{l}}-(l-1))^{n}\right]\mathcal{E}^{\beta}_{a_{l}}\right)\\
%&+\left(\sum_{n=0}^{\lceil m-1\rceil}\binom{\beta/\delta}{n}\left[(1+\nu(N-2))^{\frac{\beta}{\delta}-n}-(\nu(N-2))^{\frac{\beta}{\delta}-n}\right]\nu^{n}\left[\tau_{a_{N-1}}-(N-2)\right]^{n}\right)\mathcal{E}^{\beta}_{a_{N-1}}
%\end{aligned}
%\end{equation}
%for $\lfloor m\rfloor\delta<\beta\leq m\delta$ $(m\geq1)$, where $0\leq\nu\leq1$ and $\tau_{a_{l}}=\frac{\mathcal{E}^{\delta}_{a_{1}}+\mathcal{E}^{\delta}_{a_{2}}+\cdots+\mathcal{E}^{\delta}_{a_{l-1}}}{\mathcal{E}^{\delta}_{a_{l}}},l=2,3,\cdots,N-1$ $(N\geq4)$.
\end{thm}
\begin{proof} By the $\delta$th-polygamy \eqref{e:chap1-ineq2},
$\mathcal{E}^{\gamma}_{a_{A|B_{1}\cdots B_{N-1}}}\leq \mathcal{E}^{\gamma}_{a_{1}}+\mathcal{E}^{\gamma}_{a_{2}}+\cdots+\mathcal{E}^{\gamma}_{a_{N-1}}
$.
Using the ordering condition $\mathcal{E}^{\gamma}_{a_{i}}\geq \mathcal{E}^{\gamma}_{a_{i+1}}>0$ for $i=1,\cdots,N-2$, there exists $0\leq\nu\leq1$ such that $\mathcal{E}^{\gamma}_{a_{A|B_{1}\cdots B_{N-1}}}\leq \nu\left(\mathcal{E}^{\gamma}_{a_{1}}+\cdots+\mathcal{E}^{\gamma}_{a_{N-2}}\right)+\mathcal{E}^{\gamma}_{a_{N-1}}
$. So \eqref{e:chap4-ineq4} drived by Theorem \ref{t:chap4-thm1} and Lemma \ref{l:chap4-lem1}.
\end{proof}

%\begin{rem}\label{r:chap4-rem1} Let $\rho_{AB_{1}\cdots B_{N-1}}$ be a quantum state.

%{\bf(1)} If $\nu=1,$ then inequalities \eqref{e:chap4-ineq4} reduce to \eqref{e:chap4-ineq1}.

%{\bf(2)} If $0<\nu<1,$ then the upper bound in \eqref{e:chap4-ineq4} is tighter than that of \eqref{e:chap4-ineq1}.
%\end{rem}

%\textcolor{red}
{\textbf{Physical Interpretation.}Theorem \ref{t:chap4-thm2} combines the hierarchy parameter $m$ and the state-dependent parameter $\nu$, which quantifies the asymmetry in the distribution of assisted entanglement. When $0 < \nu < 1$, the bound in \eqref{e:chap4-ineq4} is tighter than the general bound in \eqref{e:chap4-ineq1}, reflecting that the shareability of entanglement is often constrained by the weakest link. This state-adaptive bound provides more realistic constraints for quantum communication tasks that rely on entanglement assistance, such as remote state preparation or distributed quantum computing.}

\bigskip
%Recall that the negativity of bipartite state $\rho_{A B}$ is defined by \cite{VW}: $\mathcal{N}\left(\rho_{AB}\right)=\left\|\rho_{AB}^{T_{A}}\right\|-1$, where $T_{A}$ refers to the partial transposition and $\|X\|=\operatorname{Tr} \sqrt{X X^{\dagger}}$ is the trace norm of $X$. For a mixed state $\rho_{AB}$, the convex-roof extended negativity (CREN) is defined by $\mathcal{N}_{ c}\left(\rho_{A B}\right)=\min \sum_i p_i \mathcal{N}\left(\left|\psi_i\right\rangle_{AB}\right)$ where the minimum is taken over all possible pure state decompositions $\left\{p_i,\left|\psi_i\right\rangle_{AB}\right\}$ of $\rho_{AB}$. The convex-roof extended negativity of assistance (CRENoA) is then defined by $\mathcal{N}_{a}\left(\rho_{AB}\right)=\max \sum_i p_i \mathcal{N}\left(\left|\psi_i\right\rangle_{AB}\right)$, where the maximum is taken over all possible pure state decompositions $\left\{p_i,\left|\psi_i\right\rangle_{AB}\right\}$ of $\rho_{AB}$. For convenience, we will abbreviate: $\mathcal{N}_{aAB}=\mathcal{N}_{a}\left(\rho_{AB}\right)$.

%For any any $N$-qubit pure state $|\psi\rangle_{AB_{1}\cdots B_{N-1}}$, \cite{KDS2}
%\begin{equation}\label{e:chap4-ineq4}
%\begin{aligned}
%\mathcal{N}^{2}(|\psi\rangle_{A|B_{1}\cdots B_{N-1}})\leq \sum_{i=1}^{N-1} \mathcal{N}_{a}^{2}(\rho_{AB_{i}})
%\end{aligned}
%\end{equation}
 Now consider the concurrence of assistance (CoA) as an illustrative case. CoA is defined by \cite{OV1}
 \begin{align*}
C_{a}\left(\rho_{A B}\right)=\max _{\left\{p_i,\left|\psi_i\right\rangle\right\}} \sum_i p_i C\left(\left|\psi_i\right\rangle\right),
\end{align*}
where the maximum is taken over all possible convex roofs of pure state decompositions: $\rho_{A B}=\sum_i p_i\left|\psi_i\right\rangle\left\langle\psi_i\right|$ with $p_i \geqslant 0, \sum_i p_i=1$ and $\left|\psi_i\right\rangle \in \mathcal{H}_A \otimes \mathcal{H}_B$.

For an $N$-partite pure state $|{\psi\rangle}_{AB_1\cdots B_{N-1}}$, the concurrence with respect to the partition $A|B_1\cdots B_{N-1}$ satisfies
the polygamy relation \cite{GBS}:
\begin{equation}\label{e:chap4-ineq6}
C^{2}(|\psi\rangle_{A|B_{1}\cdots B_{N-1}})\leq C^{2}_{a_{AB_{1}}}+C^{2}_{a_{AB_{2}}}+\cdots C^{2}_{a_{AB_{N-1}}}.
\end{equation}

\begin{cor}\label{c:chap4-cor1}
 Let $|\psi\rangle_{AB_{1}\cdots B_{N-1}}$ be any $N$-qubit pure state and $C_{a}$ be bipartite assisted quantum measure CoA satisfying the polygamy relation \eqref{e:chap4-ineq6}. Arrange $\{C_{a_{i}}=C_{aAB_{i'}}|i=1,\cdots,N-1\}$ in descending order, and that $C^{\delta}_{a_{i}}\geq C^{\delta}_{a_{i+1}}>0$ for $i=1,\cdots,N-2$. Then for $2(m-1)<\beta\leq 2m$ with
$m_{\in\mathbb{Z}_{+}}\geq1$, we have
\begin{equation}\label{e:chap4-ineq7}
\begin{aligned}
&C^{\beta}_{aA|B_{1}\cdots B_{N-1}} \leq \nu^{\frac{\beta}{2}}\left( C_{a_{1}}^{\beta}+\sum_{l=2}^{N-2}\left[\sum_{n=0}^{m-1}\binom{\beta/2}{n}(l^{\frac{\beta}{2}-n}-(l-1)^{\frac{\beta}{2}-n})(\tau_{a_{l}}-(l-1))^{n}\right]C^{\beta}_{a_{l}}\right)\\
&+\left(\sum_{n=0}^{m-1}\binom{\beta/2}{n}\left[(1+\nu(N-2))^{\frac{\beta}{2}-n}-(\nu(N-2))^{\frac{\beta}{2}-n}\right]\nu^{n}\left[\tau_{a_{N-1}}-(N-2)\right]^{n}\right)C^{\beta}_{a_{N-1}},
\end{aligned}
\end{equation}
where $0\leq\nu\leq1$ and  $\tau_{a_l} = \frac{\sum_{i=1}^{l-1} C_{a_i}^2}{C_{a_l}^2}$ for $l = 2, \dots, N-1$.
\end{cor}

\begin{exmp}
Let $\rho=|\psi\rangle\langle\psi|$ be the three-qubit state defined by \cite{AACJLT}:
$$
|\psi\rangle=\lambda_{0}|000\rangle+\lambda_{1}e^{i \varphi}|100\rangle+\lambda_{2}|101\rangle+\lambda_{3}|110\rangle+\lambda_{4}|111\rangle.
$$
where $\sum_{i=0}^4 \lambda_i^2=1$, and $\lambda_{i}\geq 0$ for $i=0,1,2,3,4$.
Then  $C(|\psi\rangle_{A|BC})=2 \lambda_0 \sqrt{\lambda_2^2+\lambda_3^2+\lambda_4^2}$, $C_{a_{AB}}=2 \lambda_0 \sqrt{\lambda^{2}_{2}+\lambda^{2}_{4}}$, and $C_{a_{AC}}=2 \lambda_0 \sqrt{\lambda^{2}_{3}+\lambda^{2}_{4}}$.
Set $\lambda_0=\frac{1}{9}, \lambda_1=0, \lambda_2=\frac{2}{9}, \lambda_3=\frac{2\sqrt{10}}{9}, \lambda_4=\frac{2}{3}$.
Then we have $C_{a_{A\mid BC}}=\frac{8\sqrt{5}}{81}, C_{a_{AB }}=\frac{4\sqrt{10}}{81}, C_{a_{AC}}=\frac{4\sqrt{19}}{81}$. Thus, by Lemma \ref{l:chap4-lem1}, $\nu=\frac{C^{2}_{a_{A\mid BC}}-C^{2}_{a_{AB}}}{C^{2}_{a_{AC}}}=\frac{10}{19}$. Set $m=2$ (since $m\geq 1$).
%{\color{red} Can we also draw piecewise graph in this case?}

\bigskip
By Corollary \ref{c:chap4-cor1}, for any $\delta<\beta\leq 2\delta$, the RHS of the $\beta$th-polygamy relation is
\begin{align*}
T_1&=\nu^{\frac{\beta}{2}}C_{aAC}^{\beta}+\left((1+\nu)^{\frac{\beta}{2}}-\nu^{\frac{\beta}{2}}+\frac{\beta}{2}\left[(1+\nu)^{\frac{\beta}{2}-1}-\nu^{\frac{\beta}{2}-1}\right]\nu\left(\frac{C_{aAC}^{2}}{C_{aAB}^{2}}-1\right)\right)C_{aAB}^{\beta}\\
&=\left(\frac{10}{19}\right)^{\frac{\beta}{2}}\left(\frac{4\sqrt{19}}{81}\right)^{\beta}+\left(\left(\frac{29}{19}\right)^{\frac{\beta}{2}}-\left(\frac{10}{19}\right)^{\frac{\beta}{2}}+\frac{9}{38}\beta\left[\left(\frac{29}{19}\right)^{\frac{\beta}{2}-1}-\left(\frac{10}{19}\right)^{\frac{\beta}{2}-1}\right]\right)\left(\frac{4\sqrt{10}}{81}\right)^{\beta}
\end{align*}

The following upper bound from Theorem \ref{t:chap4-thm1} is
 \begin{align*}
T_2&=C_{aAC}^{\beta}+\left(2^{\frac{\beta}{2}}-1+\frac{\beta}{2}\left(2^{\frac{\beta}{2}-1}-1\right)\left(\frac{C_{aAC}^{2}}{C_{aAB}^{2}}-1\right)\right)C_{aAB}^{\beta}\\
&=\left(\frac{4\sqrt{19}}{81}\right)^{\beta}+\left(2^{\frac{\beta}{2}}-1+\frac{9}{20}\beta\left(2^{\frac{\beta}{2}-1}-1\right)\right)\left(\frac{4\sqrt{10}}{81}\right)^{\beta}
\end{align*}
\end{exmp}

\begin{figure}[H]
    \centering
    \includegraphics[width=0.8\textwidth]{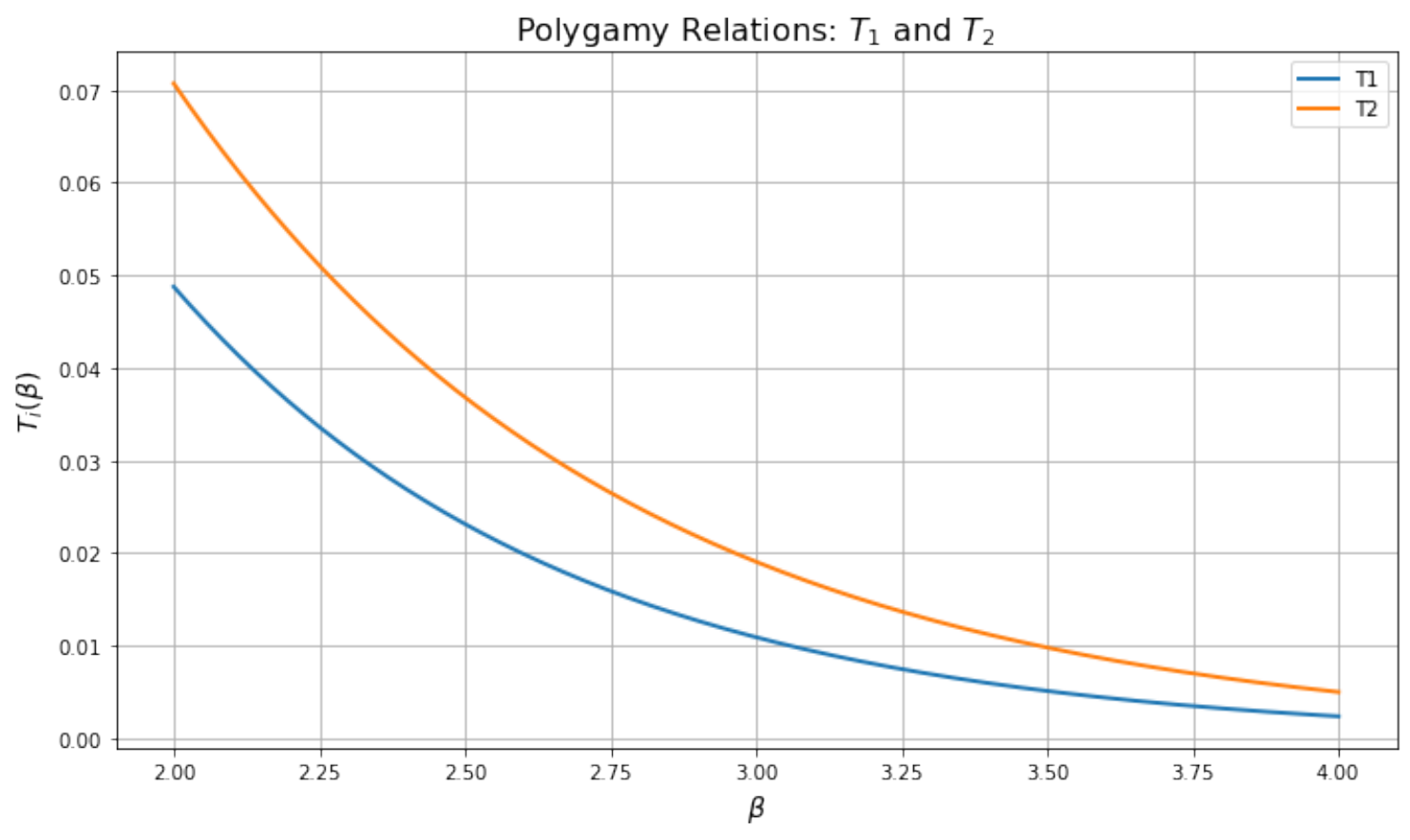}
    \caption{Comparisons between two upper bounds $T_1$ and $T_2$ associated with the $\beta$-th power polygamy relation of the concurrence of assistance, under the specific three-qubit quantum state $|\psi\rangle$ constructed in Example 4.6. The bound $T_1$ corresponds to the newly proposed inequality in Corollary 4.5, which incorporates the state-dependent refinement via the parameter $\nu = \frac{10}{19}$, while $T_2$ is derived from the classical bound in Theorem 4.1 that does not utilize this refinement. As can be seen, the curve representing $T_1$ consistently lies below that of $T_2$ for the entire domain $\beta \in [2,4]$, demonstrating that the new bound obtained through the parameterized inequality structure yields a strictly tighter constraint on the assisted entanglement distribution. %\textcolor{red}
    {It demonstrates that the new bounds obtained through a parameterized inequality structure impose tighter constraints on the distribution of assisted entanglement. This highlights the practical advantage of using state-adaptive parameters to capture the inherent asymmetry in the shareability of entanglement.}}
    \label{fig:T1T2}
\end{figure}

\section{\textbf{Appendix}}\label{s:appendix}
\subsection{\textbf{Proof of Lemma \ref{l:chap2-lem1}}}
 Consider the function:
 $$f(x,t)=(1+t)^{x}-t^{x}$$
defined on $(x,t)\in [m,+\infty)\times[k,+\infty).$
Then $\frac{\partial^{n} f(x,t)}{\partial t^{n}}=x^{\underline{n}}[(1+t)^{x-n}-t^{x-n}], (n=0,1,2,\cdots),$
where the falling factorial $x^{\underline{n}}=x(x-1)\cdots (x-n+1)$,$(n=1,2,\cdots)$, and $x^{\underline{0}}=1$.

{\bf(1)} For $x\geq m$, $\frac{\partial^{\lfloor m\rfloor} f(x,\xi_{m})}{\partial t^{\lfloor m\rfloor}}=x^{\underline{\lfloor m\rfloor}}\left[(1+\xi_{m})^{x-\lfloor m\rfloor}-\xi_{m}^{x-\lfloor m\rfloor}\right]\geq 0.$ Thus, for $ t\geq k>0$, the Taylor formula of $f(x, t)$ at the point $(x,k)$ with the Lagrange remainder term is
\begin{align*}
f(x,t)&=\sum_{n=0}^{\lfloor m-1\rfloor}\binom{x}{n}\left[(1+k)^{x-n}-k^{x-n}\right](t-k)^{n}+\binom{x}{\lfloor m\rfloor}\left[(1+\xi_{m})^{x-\lfloor m\rfloor}-\xi_{m}^{x-\lfloor m\rfloor}\right](t-k)^{\lfloor m\rfloor}\\
&\geq\sum_{n=0}^{\lfloor m-1\rfloor}\binom{x}{n}\left[(1+k)^{x-n}-k^{x-n}\right](t-k)^{n}
\end{align*}
where $\xi_{m}\in (k,t)$, $ (m\geq1, n=0,1,\ldots,\lfloor m-1\rfloor).$

{\bf(2)} For $m-1 <x\leq m$, we need to analyze this in two separate cases:

${\bf1^{\circ}}$ If $m-1 <x\leq \lfloor m\rfloor$, then $\frac{\partial^{\lfloor m\rfloor} f(x,\eta_{m})}{\partial t^{\lfloor m\rfloor}}=x^{\underline{\lfloor m\rfloor}}\left[(1+\eta_{m})^{x-\lfloor m\rfloor}-\eta_{m}^{x-\lfloor m\rfloor}\right]\leq 0.$ Thus, for $ t\geq k>0$, we have
\begin{align*}
f(x,t)&=\sum_{n=0}^{\lfloor m-1\rfloor}\binom{x}{n}\left[(1+k)^{x-n}-k^{x-n}\right](t-k)^{n}+\binom{x}{\lfloor m\rfloor}\left[(1+\eta_{m})^{x-\lfloor m\rfloor}-\eta_{m}^{x-\lfloor m\rfloor}\right](t-k)^{\lfloor m\rfloor}\\
&\leq\sum_{n=0}^{\lfloor m-1\rfloor}\binom{x}{n}\left[(1+k)^{x-n}-k^{x-n}\right](t-k)^{n}
\end{align*}
where $\eta_{m}\in (k,t)$.
%$ (m\geq1, n=0,1,\ldots,\lfloor m-1\rfloor).$

${\bf2^{\circ}}$ If $\lfloor m\rfloor <x\leq m$, then
$\frac{\partial^{\lceil m\rceil} f(x,\zeta_{m})}{\partial y^{\lceil m\rceil}}=x^{\underline{\lceil m\rceil}}\left[(1+\zeta_{m})^{x-\lceil m\rceil}-\zeta_{m}^{x-\lceil m\rceil}\right]\leq 0.$
Therefore, for $ t\geq k>0$
\begin{align*}
f(x,t)&=\sum_{n=0}^{\lceil m-1\rceil}\binom{x}{n}\left[(1+k)^{x-n}-k^{x-n}\right](t-k)^{n}+\binom{x}{\lceil m\rceil}\left[(1+\zeta_{m})^{x-\lceil m\rceil}-\zeta_{m}^{x-\lceil m\rceil}\right](t-k)^{\lceil m\rceil}\\
&\leq\sum_{n=0}^{\lceil m-1\rceil}\binom{x}{n}\left[(1+k)^{x-n}-k^{x-n}\right](t-k)^{n}
\end{align*}
where $\zeta_{m}\in (k,t)$.
And for any real number $a$, $\lfloor a \rfloor$ and $\lceil a \rceil$ denote the floor and ceiling functions, respectively.

\subsection{\textbf{Proof of Corollary \ref{c:chap2-cor1}}} For $t\geq k\geq1$, set $F(t)=\sum_{n=1}^{\lfloor m-1\rfloor}\binom{x}{n}[(1+k)^{x-n}-k^{x-n}](t-k)^{n}- k^{-x}+t^{-x}-(2^{m-1}-2)x(t-k).$ Then
\begin{align*}
F^{'}(t)&=\sum_{n=1}^{\lfloor m-1\rfloor}\binom{x}{n}[(1+k)^{x-n}-k^{x-n}]n(t-k)^{n-1}-xt^{-x-1}-(2^{m-1}-2)x \\ F^{''}(t)&=\sum_{n=2}^{\lfloor m-1\rfloor}\binom{x}{n}[(1+k)^{x-n}-k^{x-n}]n(n-1)(t-k)^{n-2}+x(x+1)t^{-x-2}\geq 0
\end{align*}
 So
\begin{align*}
F^{'}(t)\geq F^{'}(k)&=x[(1+k)^{x-1}-k^{x-1}-k^{-x-1}-(2^{m-1}-2)]\\
&\geq x[(1+k)^{x-1}-k^{x-1}-(2^{m-1}-1)]\quad\quad(\text{by $k\geq 1$ })\\
&\geq 0 \quad\quad(\text{by $h(x,k)\geq h(m,1)$ })
\end{align*}
where the function $h(x,t)=(1+t)^{x-1}-t^{x-1}$ is increasing in $x\geq m\geq2,t\geq k\geq 1$. This implies $F(t)$ is increasing function for $t\geq k\geq1$, so $F(t)\geq F(k)=0$, which yields \eqref{e:chap2-ineq4}.

\subsection{\textbf{Proof of Lemma \ref{l:chap2-lem2}}}

We use induction on $N$. The case of $N = 1$ is clear. Assume \eqref{e:chap2-ineq7} holds for $<N$. For given $p_i$ it is clear that $p_{1}+p_{2}+\cdots+p_{N-1}\geq (N-1)p_{N}$. Using Lemma \ref{l:chap2-lem1} we have that

\begin{align*}
&\left(\sum_{i=1}^{N}p_{i}\right)^{x}=\left(p_{1}+p_{2}+\cdots+p_{N}\right)^{x}=p_{N}^{x}\left(1+\frac{p_{1}+\cdots+p_{N-1}}{p_{N}}\right)^{x}\\
&\geq p_{N}^{x}\left(\left(\frac{p_{1}+\cdots+p_{N-1}}{p_{N}}\right)^{x}+\sum_{n=0}^{\lfloor m-1\rfloor}\binom{x}{n}[(N^{x-n}-(N-1)^{x-n})]\left(\tau_{N}-(N-1)\right)^{n}\right)\\
&=(p_{1}+\cdots+p_{N-1})^{x}+p_{N}^{x}\sum_{n=0}^{\lfloor m-1\rfloor}\binom{x}{n}[(N^{x-n}-(N-1)^{x-n})]\left(\tau_{N}-(N-1)\right)^{n}
\end{align*}
where $\tau_{N}=\frac{p_{1}+\cdots+p_{N-1}}{p_{N}}.$
By the inductive hypothesis, the above is no less than the right-hand side (RHS) of \eqref{e:chap2-ineq7}. The proof of \eqref{e:chap2-ineq8} is similar.

\section{\textbf{Conclusion}}

In this work, we have developed a unified and parameterized framework for refining monogamy and polygamy inequalities of bipartite entanglement measures in multipartite quantum systems. By dividing the region through an integer parameter $m \geq 1$, we have established:

\begin{itemize}
    \item A hierarchy of tighter $\alpha$-power monogamy inequalities for any bipartite entanglement measure $\mathcal{E}$, valid for $\alpha \geq m\gamma$;
    \item A family of refined $\beta$-power polygamy inequalities for any assisted entanglement measure $\mathcal{E}_a$, applicable for $(m-1)\delta < \beta \leq m\delta$.
\end{itemize}
%\textcolor{red}
{This framework is physically motivated by the inherently hierarchical and asymmetric nature of entanglement distribution. The parameter $m$ controls the depth of considered correlations, while state-specific parameters $\mu$ and $\nu$ quantify entanglement imbalance, allowing our bounds to adapt to the actual geometry of quantum states.}

%\textcolor{red}
{The resulting monogamy bound emerges as an optimal piecewise function of $\alpha$, with segments activated at thresholds $2m$. This structure incorporates higher-order corrections, delivering strictly sharper constraints than previous continuous formulations. Our results generalize and strengthen existing relations (e.g.,~\cite{CJMW1, GYG}), naturally recovering standard bounds when $m=1$. Analytical proofs and numerical tests using concurrence and concurrence of assistance confirm their enhanced tightness.}

%\textcolor{red}
{By further introducing state-dependent parameters $\mu$ and $\nu$, we also tighten the relations in~\cite{CJMW2}. This methodology advances the quantitative understanding of entanglement distribution and provides practical tools for quantum communication, networks, and multipartite information processing.}

\bigskip
\centerline{\bf Acknowledgments}
\medskip

The work is supported in part by 
%the National Natural Science Foundation of China (grant nos.12171303 and 12001218), 
the Simons Foundation (grant no. MP-TSM-00002518).

\bibliographystyle{plain}

\end{document}